\documentclass[11pt, letterpaper]{article}
\usepackage[margin=1.1in,footskip=0.25in]{geometry}

\usepackage{amsmath, amssymb}
\usepackage{amsthm}
\usepackage{bm}
\allowdisplaybreaks

\usepackage[T1]{fontenc}
\usepackage[scaled=0.8]{beramono}

\let\counterwithin\relax
\usepackage{chngcntr}

\usepackage{xcolor}
\definecolor{linkcolor}{HTML}{0645AD}
\usepackage{hyperref}
\hypersetup{
     colorlinks   = true,
     allcolors    = linkcolor,
}

\definecolor{white}{RGB}{255,255,255}
\definecolor{crimson}{RGB}{220,20,60}
\definecolor{blue}{RGB}{0,0,205}
\definecolor{myblue}{RGB}{80,80,160}
\definecolor{mygreen}{RGB}{80,160,80}
\colorlet{tcrimson}{white!40!crimson}
\colorlet{tblue}{white!40!blue}
\definecolor{porange}{HTML}{EE7F2D}

\def\*#1{\mathbf{#1}}

\newcommand{\E}{\mathbb{E}}

\newcommand{\indep}{\!\perp\!\!\!\perp}

\usepackage{chngcntr}
\usepackage{apptools}
\AtAppendix{\counterwithin{lemma}{section}}

\theoremstyle{definition}
\newtheorem{proposition}{Proposition}

\newtheorem{corollary}{Corollary}
\newtheorem{lemma}{Lemma}
\newtheorem{assumption}{Assumption}

\usepackage{setspace}
\usepackage{pdfpages}
\usepackage{booktabs}
\usepackage{placeins} 

\usepackage{tikz}

\usepackage{natbib}
\usepackage{tabularx}
\newcolumntype{R}{>{\raggedleft\arraybackslash}X}
\newcolumntype{C}{>{\centering\arraybackslash}X}
\newcolumntype{L}{>{\raggedright\arraybackslash}X}

\usepackage[margin=35pt,
            labelfont=bf 
            ]{caption}
\usepackage{marginnote}

\begin{document}
\title{\textbf{Synthetic Area Weighting for \\ Measuring Public Opinion in Small Areas}\thanks{We thank
Eli Ben-Michael, Avi Feller, Max Goplerud, Kosuke Imai, Hiroto Katsumata, Lauren Kennedy, Chris Kenny, Jonathan Robinson, David Shor, Yajuan Si, Arjun Vishwanath, Shun Yamaya, members of the Imai Research Group, and participants of the 2020 MRP conference at Columbia University for comments and suggestions.
A \textsf{R} package, \texttt{synthArea},  is available to implement the proposed method at \url{https://github.com/kuriwaki/synthArea}.}}
\author{
Shiro Kuriwaki\thanks{PhD candidate, Department of Government, Harvard University. URL: \url{https://www.shirokuriwaki.com/}.}\\
{\small Harvard University}\\
\and
Soichiro Yamauchi\thanks{PhD candidate, Department of Government, Harvard University. URL: \url{https://soichiroy.github.io/}.}\\
{\small Harvard University}\\
}
\date{\small\today}
\maketitle

\setstretch{1.15}


\begin{abstract}
The comparison of subnational areas is ubiquitous but survey samples of these areas are often biased or prohibitively small. 
Researchers turn to methods such as multilevel regression and poststratiﬁcation (MRP) to improve the efficiency of estimates by partially pooling data across areas via random effects. 
However, the random effect approach can pool observations only through area-level aggregates.
We instead propose a weighting estimator, the \emph{synthetic area estimator}, which weights on variables measured only in the survey to partially pool observations \emph{individually}. The proposed method consists of two-step weighting: first to adjust differences across areas and then to adjust for differences between the sample and population. 
Unlike MRP, our estimator can directly use the national weights that are often estimated from pollsters using proprietary information.
Our approach also clarifies the assumptions needed for valid partial pooling, without imposing an outcome model. 
We apply the proposed method to estimate the support for immigration policies at the congressional district level in Florida. Our empirical results show that small area estimation models with insufficient covariates can mask opinion heterogeneities across districts.

\end{abstract}

\clearpage

\setstretch{1.3}

\section{Introduction}


The comparison of states, districts, and counties in social science increasingly use survey samples to estimate behavior at subnational units.
For example, scholars have estimated public opinion on the minimum wage,
gay rights, the death penalty,
or anti-smoking regulation \citep{simonovits2019responsiveness, hansen2015symbolic,canes2014judicial, pacheco2012social},
whether that is at the level of states,
congressional districts
or municipal governments and regions \citep[e.g.,][]{lax2012democratic,howe2015geographic,hertel2019legislative, alesina2002trusts,tausanovitch2014representation,putnam2007pluribus}.
The necessary step (or sometimes the main goal) in these studies is the estimation of population opinion in each small area using small samples of survey respondents.

%
Despite its importance in empirical studies,
estimating public opinion at the sub-population level is challenging. Estimates suffer from high variance due to small samples,
and the survey may be unrepresentative of the population as-is.
A popular solution to this problem is multilevel regression and the post-stratification (MRP) \citep{gelman1997poststratification}.
MRP combines poststratification, which addresses bias due to non-representative sampling, with
traditional small-area estimation, which addresses high variance due to small samples \citep[e.g.,][]{fay1979estimates}.
It allows scholars to pool observations across sub-populations (or small areas) while accounting for heterogeneities via random effects.
The key assumption behind the random effect approach including MRP is that area-level unobservables are independent of the outcome conditional on two types of covariates:
area-level \emph{aggregate} as well as individual-level covariates that are observed both in the population and survey.
In other words, the assumption will not be met when there is an unmeasured individual-level covariate that is associated with the outcome and area-level unobservables,
which leads to biased estimates due to pooling heterogenous areas.

In this paper, we propose a weighting-based approach to small area estimation, which we call the \emph{synthetic area estimator},
that utilizes individual-level covariates for partial pooling.
Specifically, the proposed method enables scholars to incorporate individual-level covariates that are \emph{only observed in the survey}.
While the method still requires independence of the area-level unobservables with the outcome so that we can borrow information across small areas,
it can condition on more variables that the existing random effect approach has not utilized, which makes the independence assumption more plausible.
In addition, we derive a testable condition of the independence assumption in our framework so that researchers can assess if their partial pooling strategy is plausible with the data.

One additional attractive feature of the proposed method is that  analysts can directly use the national post-stratification (calibration) weights that come with most surveys, because the synthetic area estimator can be written as a transformation of this national weights.
These weights provided by survey firms are not calibrated to the small areas, but still encode useful information about the entire population and are excellent candidates for justifying sampling ignorability. While most MRP methods requires that researchers jettison these weights and create new post-stratification weights through a rather laborious process incorporating fewer covariates than the survey firms' weights, the synthetic area method offers an option to use the survey firms' weights directly.

Theoretically, our approach clarifies the assumptions necessary to justify small area estimators in general.
We show that the small area quantity is nonparametrically identified under \emph{area ignorability},  which assumes independence between the outcome and area indicators given covariates, in addition to the standard  \emph{sampling ignorability}, which assumes independence between the outcome and sample selection given covariates.
The theoretical contribution of this paper is thus to show that the identification formula can incorporate survey-only variables that are not utilized in the previous literature.
In contrast to recent work on MRP, our results provide a nonparametric, as opposed to a model-based, justification to a small area estimation strategy.
Therefore, researchers who intend to analyze small areas through subgroups of surveys can interrogate or justify the validity of their approach by these two conditions.

Finally, we conduct a validation study as well as an empirical application to illustrate the proposed method, implemented via an open-source \textsf{R} software package, \textsf{synthArea}.
First we estimate congressional district level-turnout using our method and MRP using a large survey dataset, and compare it to the ground-truth levels of turnout.
We demonstrate that the proposed method predicts the outcome modestly better than standard MRP.  To the extent that the synthetic area method errs, it errs in the same way as MRP estimates.
We then show how the proposed method can sensibly heterogenous estimates for preferences for public policy, through adjusting for the variables in the survey that MRP cannot leverage.
In our example of the support for deportation of immigrants in Florida, a racially diverse state, adjusting for party identification and race in the survey produces estimates that are consistent with the population distribution of race. In contrast, estimates based only on the joint distribution available in the Census (age, gender, and education) systematically underestimate the diversity across congressional districts.


This paper contributes to a growing literature on MRP or more generally small area estimation.
The contribution of the proposed method, however, is distinct from recent developments in MRP.
While recent work such as \cite{bisbee2019barp} or \cite{ornstein2020stacked} propose a complex regularization approach to
the partial pooling strategy and improves upon the original random effect approach, none of the methods do not consider incorporating individual-level survey-only covariates to improve the validity of partial pooling (also see \cite{goplerud2018sparse,montgomery2018tree}).
The problem we consider in this paper is also different from the literature on synthetic \emph{populations}, that attempts to impute the joint distribution of sub-population targets \citep{kastellec2015polarizing,leemann2017extending,lauderdale2020model,ghitzasteitz2020}.
Our paper treats the poststratification part as given, and improves the partial pooling strategy. Therefore, these methods can be used together with synthetic area estimation to improve the quality of poststratification.

Our method is also closely connected to the literature
on the subgroup-specific unobservables in causal inference.
In the context of accounting for unobserved cluster-level heterogeneity,
\cite{arkhangelsky2018role} derive a similar condition to the area ignorability
and show that the empirical summary of each cluster can capture the unobserved heterogeneity across clusters.
However, their target estimand is the average treatment effect (or the population mean), and therefore their method is not directly applicable to estimate small area (i.e., cluster-specific) quantities.
\cite{ben2020varying} propose a weighting-based approach to compute the subgroup quantities, which partially pools observations to estimate weights that are subgroup specific \citep[also see][]{dong2020subgroup}. Although their weights have a smaller variance, the target quantities still suffer from the small sample problem because the method uses only observations from the specific subgroup for the final estimation,
whereas our method utilizes all observations with an additional assumption.

\section{The Small Area Estimation Problem}\label{sec:sae}


Researchers start with a survey sample of $n$ observations taken from an infinite population.
Let $S \in \{0, 1\}$ denote the survey inclusion indicator
that takes $1$ if a unit is selected into the survey
and $0$ otherwise.
We are interested in units living in sub-areas of the population
defined by the categorical variable $A \in \{1, \ldots, J\}$.
In the context of our validation study,
$A_{i}$ indicates the congressional district the $i$th unit resides in.
In the population, we observe covariates which we denote by $\*X^{P}$ (e.g., basic population demographics given by Census).
In addition to those variables, the survey measures the outcome $Y$,
and additional covariates not observed in the population (e.g., interest in politics or partisan self-identification)
which we denote by $\*X^{S}$. The distinction between population and survey-only covariates will be important to illustrate the difference between MRP and our proposed method. We summarize this notation in Table \ref{tab:X-types}.

\begin{table}[t]
\centering
\small
\caption{\textbf{Types of Covariates in Small Area Estimation}\label{tab:X-types}}
\begin{tabularx}{0.9\linewidth}{cccCC}
\toprule
& & & \multicolumn{2}{c}{Incorporated during Partial Pooling}\\\cmidrule(lr){4-5}
& & Examples&  MRP & synthArea\\
\midrule
Population Variables & $\*X^{P}_i$ & Age, Gender, Race & $\checkmark$ & $\checkmark$\\
Survey-only Variables & $\*X^{S}_i$ & Party ID, News Interst & & $\checkmark$\\
\bottomrule
\end{tabularx}
\bigskip

\footnotesize
\begin{minipage}{0.9\linewidth}
\emph{Note}: We distinguish between two sets of individual variables in this study. Population variables $\*X^{P}$ (also referred to as Census variables) are those for which popualtion quantities are observed. Survey-only varibales $\*X^{S}$ are those only measured in the survey sample. Within each type of variable, variables must be observed jointly (e.g. the joint distribution of age, gender, and race in the population). Variables for which only aggregate marginal distributions are available are denoted as $\overline{\*X}^{A}_{j}.$
\end{minipage}
\end{table}

We are interested in estimating population mean of the outcome in the target area,
\begin{equation}\label{eq:estimand}
\tau_{j} = \E[Y \mid A = j]
\end{equation}
for area $j = 1, \ldots, J$.
We denote a vector of area-specific population means by $\bm{\tau} = (\tau_{1}, \tau_{2}, \ldots, \tau_{J})^{\top}$.
In our example, $\tau_{j}$ corresponds to turnout at the congressional district level.

\subsection{The Direct Estimator}\label{subsec:direct}

In standard (i.e. non-small area) survey inference,
we attempt to estimate $\tau_{j}$ using a sample from the target area.
This strategy is justified when the selection into the survey is independent of the outcome of interest,
possibly given observed characteristics.
Assumption~\ref{assump:sampling-ignorability} formally states the assumption,
\begin{assumption}[Sampling ignorability]\label{assump:sampling-ignorability}
  The following two conditions hold:
  \begin{enumerate}
    \item[(a)] Conditional independence: $Y \indep S \mid  \*X^{P}, A$.
    \item[(b)] Overlap: $0 < \Pr(S = 1 \mid \*X^{P}, A) < 1$ for all $\*X^{P} \in \mathcal{X}^{P}$.
  \end{enumerate}
\end{assumption}

Assumption~\ref{assump:sampling-ignorability} (a) states that
among residents in the target area $(A_{i} = j)$,
the selection into the survey
is independent of the outcomes of interest given covariates $\*X^{P}_{i}$.
Conditioning on the area indicator $A$ allows for the presence of area-specific unobservables that
are correlated with the outcome and the selection.
This is a fundamental assumption that is required for most of survey inference,
though recent studies show that its validity
is sometimes questionable \citep[e.g.,][]{bailey2017designing,meng2018statistical}.
Assumption~\ref{assump:sampling-ignorability} (b)
states that all covariate profiles have a positive probability of being sampled into the survey.
This assumption might be violated if the survey cannot sample some particular group of people
in the population.

Under Assumption~\ref{assump:sampling-ignorability},
we can nonparametrically identify $\tau_{j}$ by using
survey observations from the target area,
\begin{equation}\label{eq:identification-direct}
\tau_{j} = \E\bigg\{
\frac{\bm{1}\{S = 1, A = j\}}{\Pr(A = j)}\frac{Y}{\pi_{j}(\*X)}
\bigg\}
\end{equation}
where  $\pi_{j}(\*X_{i}) \equiv \Pr(S_{i} = 1 \mid \*X^{P}_{i}, A_{i} = j)$
is the sampling probability for area $j$.
The estimator is then given as the sample analog of the above expression.
Following the convention in small area estimation \citep{Ghosh2020}, we refer this estimator as the \emph{direct estimator} because it directly uses observations from the area of interest.

We can express the direct estimator $\widehat{\tau}^{\texttt{dir}}$ as a weighted average of the following inverse probability weighting (IPW) form,
\begin{equation}\label{est:direct}
\widehat{\tau}^{\texttt{dir}}_{j} = \sum_{i = 1}^{n}\widehat{w}^{\texttt{dir}}_{ij} Y_{i},
\quad \text{where}\quad
\widehat{w}^{\texttt{dir}}_{ij} \propto
\bm{1}\{A_{i} = j\} \frac{1}{\widehat{\pi}_{j}(\*X_{i})}  +
\bm{1}\{A_{i} \neq j\} \cdot 0.
\end{equation}
The above expression illustrates that the direct estimator uses  observations from the target area, and  places the weight of zero on observations in the other areas, $A_{i} \neq j$.

Even with Assumption~\ref{assump:sampling-ignorability} which gives unbiasedness, the limitation of the direct estimator is its high variance. The sample size of $A_{i} = j$ is often small. Using a standard formula for a standard error for a proportion, a direct estimator for an area with $n_j = 100$ has a margin of error of up to 10 percentage points assuming a simple random sample.  Upon applying post-stratification weights, the standard error of the estimator becomes even larger.

An additional complication is that the weights $\pi_{j}(\*X)$ are specific to the area (as  indicated by the subscript $j$), so they have to be computed for each area separately.
As the estimate of weights is itself a function of a small number of observations,
this further contributes to the large variability in the final estimate.

\subsection{Small Area Estimators and MRP}

To overcome the issues associated with the direct estimator,
the small area estimator employs the random effect approach to ``borrow information'' from other areas, also described as partial pooling \citep{ghosh1994small,rao2014small}.
The key idea is to assign non-zero weight to observations in other areas for estimating the area-specific quantity (i.e., $\tau_{j}$), and thus address the limitations of the direct estimator.
Multilevel regression and post-stratification (MRP) combines the random effect shrinkage approach in classical small-area estimation with covariate adjustment to account for non-representative sampling \citep{gelman1997poststratification,park2004bayesian}.

To illustrate the core intuition behind the random effects approach to partial pooling, we first consider the following linear outcome model.
\begin{equation}\label{eq:mrp-outcome}
Y_{i} = \alpha_{A_{i}} + \bm{\beta}^{\top}\*X^{P}_{i} + \epsilon_{i}, \quad
\epsilon_{i} \sim \mathcal{N}(0, \sigma^{2})
\end{equation}
where $\alpha_{j}$ denotes an area-specific intercept.
The idea behind MRP does not depend on the linearity, so we keep the linear model (instead of logistic regression, which is more common in practice) for a simple exposition of the method.

Under Assumption~\ref{assump:sampling-ignorability}, we can consistently estimate parameters in the model (Equation~\ref{eq:mrp-outcome}) using survey data.
But, the estimate of the area-specific intercepts still suffers from the small-sample problem because it should be estimated based on samples from each area separately.

The random effect approach gives a consistent estimator for each area but with the cost of a new assumption. We first posit that area-specific parameter $\alpha_{j}$ is generated from a common distribution conditional on the area-level variables $\overline{\*X}^{A}_{j}$,
\begin{equation}\label{eq:mrp-area-re}
\alpha_{j} = \alpha + \bm{\gamma}^{\top}\overline{\*X}^{A}_{j} + e_{j}, \quad e_{j}\sim \mathcal{N}(0, \sigma^{2}_{\alpha}).
\end{equation}

How can this estimator be consitent for the estimand? The formulation in \ref{eq:mrp-area-re} implies that the area-specific parameter $\alpha_{j}$ takes a similar value when two areas share similar area-level profiles $\overline{\*X}^{A}_{j}$, such as the proportion White or proportion of voters above age 65.
Additional differences between two areas with similar values of $\overline{\*X}^{A}_{j}$ are attributed to a random error $e_{j}.$  The new assumption is that $e_{j}$ is independent of $\overline{\*X}^{A}_{j}$. We  will formalize this general assumption when we discuss our proposed alternative partial pooling estimator. Details on the assumptions needed for random effect models are discussed in \citet{wooldridge2010econometric} and \citet{arkhangelsky2018role}, as well as a treatment for MRP specifically by \citet{si2020use}.

The random effect approach effectively eliminates the area-specific parameter $\alpha_{j}$ from the model,
and therefore addresses the small-sample problem associated with it.
To see this in our running model, combining Equation~\eqref{eq:mrp-outcome} and \eqref{eq:mrp-area-re},
we obtain the following equation,
\begin{align*}
Y_{i} =  \alpha + \bm{\beta}^{\top}\*X^{P}_{i} + \bm{\gamma}^{\top}\overline{\*X}^{A}_{A_{i}}
 + \widetilde{\epsilon}_{i}
\end{align*}
where $\widetilde{\epsilon}_{i} = \epsilon_{i} + e_{A_{i}}$.
This equation suggests that the predicted outcome for unit $i$ and $i'$ will be similar as long as their population covariate $\*X^{P}$ and area-level variable $\overline{\*X}^{A}$ take similar values
even when they live in different areas (i.e, $A_{i} \neq A_{i'}$).

Although it is effective to address the small-sample problem,
the limitation of the random effect approach is that the partial-pooling occurs at the aggregate level, instead of a finer-grained individual level.
As Equation~\eqref{eq:mrp-area-re} indicates, the random effect models treat two areas ``similar'' when $\overline{\*X}^{A}_{j}$ is close to $\overline{\*X}^{A}_{j'}$.
However, this does not guarantee that individuals living in two areas are similar (beyond characteristics observed in $\*X^{P}$). In fact, the aggregate similarity \emph{does not} imply similarities at the individual level.
Therefore, the random effect approach could potentially pool areas that are similar at the aggregate level but heterogenous at the individual level.

\section{The Synthetic Area Estimator}\label{sec:method}


Random effects are not the only way to partially pool estimators. We submit that weighting methods can be combined to achieve the same effect, and that this approach can incorporate more covariates to make partial pooling more plausible. This section presents the proposed methodology.
In Section~\ref{sec:area-ignorability}, we formalize the assumption for partial-pooling approach which generalizes what we described in Section~\ref{sec:sae}.
Based on that assumption, we propose a new weighting based estimator that can incorporate individual-level covariates.

\subsection{The Area Ignorability Assumption for Partial Pooling}\label{sec:area-ignorability}

The challenge associated with borrowing information across areas is set up in Figure~\ref{fig:sample-space}, where we represent the simple case when there is only one area of interest.  That leads to the partition of  the data space into four quadrants by the value of $S_{i} \in \{0, 1\}$ (population and sampled) and $A_{i} \in \{0, 1\}$ (outside the target area and inside the target area).

\begin{figure}[tb]
  \centerline{
    \includegraphics[scale=0.85]{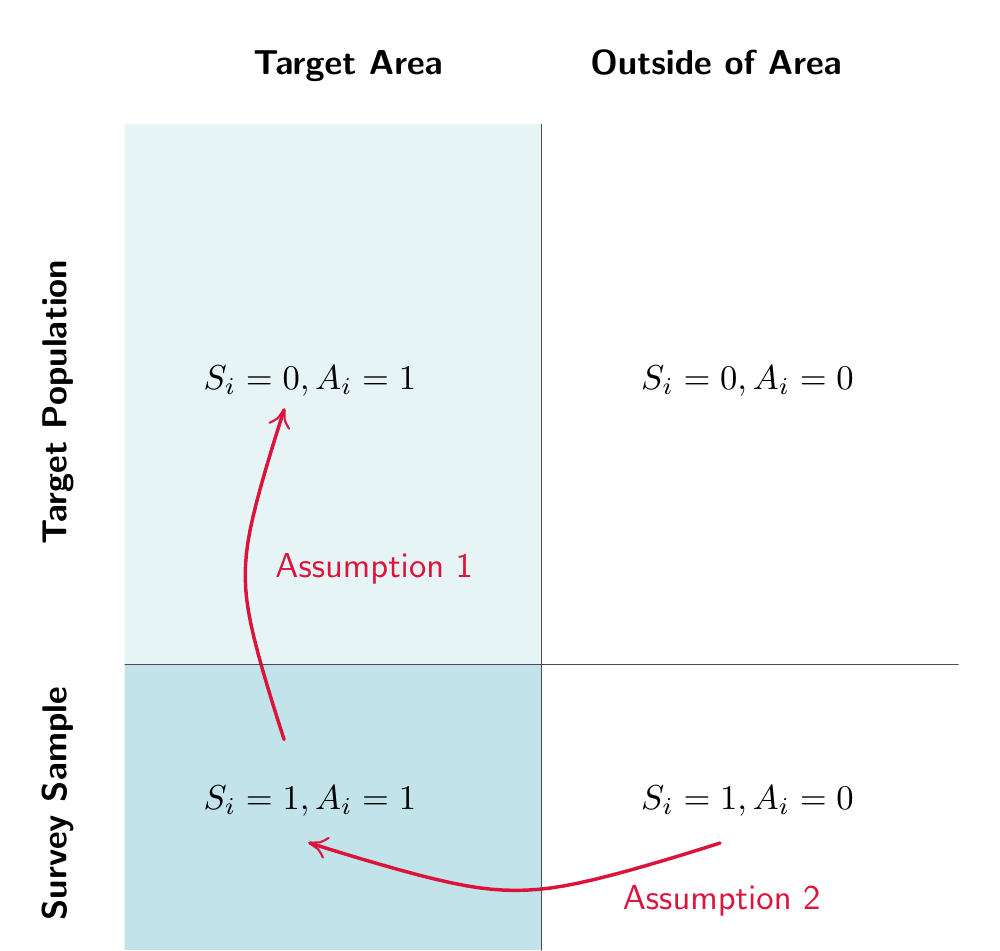}
  }
  \caption{\textbf{
    Graphical Summary of Assumption~\ref{assump:sampling-ignorability} and \ref{assump:area-ignorability}}.
    The figure represents the necessary assumptions for a weighting approach to small area estimation.  The researcher is interested in estimating the population quantity in a target (small) area, represented in the top-left quadrant. Non-small area problems are involved with adjusting the survey sample of the target area to population covariates, as shown in Assumption \ref{assump:sampling-ignorability}. Assumption \ref{assump:area-ignorability}, area ignorability, is required for valid partial pooling.
  }
  \label{fig:sample-space}
\end{figure}

The key assumption for the partial-pooling approach is that there are no area-specific unobservables that are correlated with the outcome and  sampling, which we will call \emph{area ignorability} and formally state in
Assumption~\ref{assump:area-ignorability}.
The first part of the assumption states that the area indicator $A$ is independent of the outcome
conditional on two types of covariates $\*X^{P}$ and $\*X^{S}$, among those who are likely to be sampled into the survey ($S = 1$).
In other words, we assume that area-specific unobservables do not directly affect the individual outcome after controlling for individual covariates.
The second part of the assumption is overlap.
The overlap assumption says that we can find individuals with similar observed characteristics in all areas used for the analysis.
The assumption will be violated
if, for example, districts are perfectly divided by partisanship.
\begin{assumption}[Area ignorability]\label{assump:area-ignorability}
  The following two conditions hold:
  \begin{itemize}
    \item[(a)] Conditional independence:
      $Y \indep A \mid \*X^{P}, \*X^{S}, S = 1$.
    \item[(b)] Overlap: $0 < \Pr(A = j \mid \*X^{P}, \*X^{S}, S = 1) < 1$
    for all $(\*X^{P}, \*X^{S}) \in \mathcal{X}^{P} \times \mathcal{X}^{S}$.
  \end{itemize}
\end{assumption}
A similar assumption can be given for the random effects approach described earlier in Equation \ref{eq:mrp-area-re}. The main difference is that the random effects replaces $\*X^{S}$, an individual level variable, with $\overline{\*X}^{A}$, an aggregate variable. The random effects assumption is a stronger one, because it assumes the same area ignorability with less granular covariates.

To see how the two assumptions relate to each other, it is helpful to refer to the arrows in Figure \ref{fig:sample-space}. The population in the target area is represented by the top left quadrant of the figure (shaded in light blue)
and we are interested in estimating the mean outcome in this area.
Assumption~\ref{assump:sampling-ignorability} (sampling ignorability) allows researchers to use
a sample from the same area (corresponding to the lower left quadrant; shaded in dark blue).
Assumption~\ref{assump:area-ignorability} (area ignorability) allows researchers to
utilize a sample from the non-target area (the lower right quadrant),
in a way that is comparable to the sample from the target area.

Importantly, in contrast to Assumption~\ref{assump:sampling-ignorability} (sampling ignorability), the statement in Assumption~\ref{assump:area-ignorability} is conditional on the sampling indicator, $S = 1$ and is therefore an assumption about the subpopulation that we could observe, instead of the population as a whole.
This allows us to leverage additional covariates that are measured only in the survey (i.e., $\*X^{S}_{i}$) but are not measured at the population level.
In our validation study, we utilize variables that are predictive of turnout but cannot be used in standard post-stratification weighting, such as interest in the news or marital status.

Therefore, by assuming that samples are comparable across areas given a rich set of covariates,
Assumption~\ref{assump:area-ignorability} enables researchers to estimate the quantity of interest using observations from other areas (i.e., not from the target area).
More specifically, Assumption~\ref{assump:area-ignorability} makes samples from non-target areas comparable to target area samples after conditioning on variables $\*X^{P}_{i}$
and $\*X^{S}_{i}$.
Recall that Assumption~\ref{assump:sampling-ignorability} makes the target-area sample comparable to the \textit{population} for the target area.
Thus, by combining Assumption~\ref{assump:sampling-ignorability} and \ref{assump:area-ignorability},
we can make the sample from non-target area ($A_{i} \neq j$ and $S_{i} = 1$)
comparable to the \textit{population} in the target area.

Although we can increase the validity of the assumption by conditioning on additional variables,
the conditional independence assumption is still a strong assumption.
In Section~\ref{sec:assess-independence}, we propose a procedure to assess the assumption using the observed data.

\subsection{The Synthetic Area Estimator}\label{subsec:synthArea}

Based on the two assumptions, we now provide a new identification formula for the target quantity using the partial pooling strategy that combines the direct and indirect estimator (Proposition~\ref{prop:synthArea-identification}).
\begin{proposition}[Identification of the Small Area Estimate by Pooling]\label{prop:synthArea-identification}
Under Assumption~\ref{assump:sampling-ignorability} and \ref{assump:area-ignorability}, the estimand $\tau_{j}$ is nonparametrically identified as
\begin{align*}
\tau_{j} = \E\bigg\{
\frac{\bm{1}\{S = 1\}}{\Pr(A = j)}
\frac{\Pr(A = j\mid  \*X^{P}, \*X^{S}, S = 1)}{\Pr(A = j \mid \*X^{P}, S = 1)}
\frac{\Pr(A = j \mid \*X^{P})}{\Pr(S = 1 \mid \*X^{P})}
 Y
\bigg\}
\end{align*}
for each area $j = 1, \ldots, J$.
\end{proposition}
\begin{proof}
See Appendix~\ref{sec:proofs}.
\end{proof}

Proposition~\ref{prop:synthArea-identification} shows that the population average
for the target area, $\tau_{j}$, is identified using samples from all areas.
We can see this by realizing that the identification formula does not involve the indicator function about the area, $\bm{1}\{A = j\}$, which is present in the direct estimator.


Although the above identification formulation is convenient for developing an estimator, we can obtain better intuition about the proposed method
by considering the following alternative formulation.

\begin{corollary}[Alternative Identification Formula]
\label{cor:alternative-identification}
Let $\pi_{j}(\*X) = \Pr(S = 1 \mid \*X^{P}, A = j)$ and $p_{j}(\*X) = \Pr(A = j \mid \*X^{P}, \*X^{S}, S = 1)$.
Then, the target estimand is also identified as
\begin{align*}
\tau_{j} &=
\underbrace{\E\bigg\{
\frac{\bm{1}\{S = 1, A = j \}}{\Pr(A = j)}
\frac{Y}{\pi_{j}(\*X)}\
p_{j}(\*X)
\bigg\}
}_{\text{Direct}}\\
&\quad +
\underbrace{\E\bigg\{
\frac{\bm{1}\{S = 1, A \neq j \}}{\Pr(A = j)}
\frac{\Pr(A = j \mid \*X^{P}, \*X^{S}, S = 1)}{\Pr(A \neq j \mid \*X^{P}, \*X^{S}, S = 1)}
\frac{Y}{\pi_{j}(\*X)}\
(1 - p_{j}(\*X))
\bigg\}
}_{\text{Indirect (Partially Pooled)}}.
\end{align*}
\end{corollary}

The above corollary is informative for understanding how Assumption~\ref{assump:sampling-ignorability} and Assumption~\ref{assump:area-ignorability} plays a role for identification.
The first term on the right-hand side is almost identical (except for the $p_{j}(\*X)$ term) to the formula for the direct estimator in Equation~\eqref{eq:identification-direct}.
Since this term only involves observations from the target area, (notice the indicator variable $\bm{1}\{A = j\}$), this term does not require the second assumption.
The second term is more involved than the first term in that
it has an additional term $\Pr(A = j \mid \*X^{P}, \*X^{S}, S = 1) / \Pr(A \neq j \mid \*X^{P}, \*X^{S}, S = 1)$.
This term is motivated by the second assumption (Assumption~\ref{assump:area-ignorability}) and corrects the differences between two areas (i.e., the target area and the other areas) in terms of observed variables, $\*X^{P}$ and $\*X^{S}$.
We weight observations in the other areas such that the distribution will be identical to the samples in the target area, as indicated by an arrow in Figure~\ref{fig:sample-space} that maps observations in $\{S = 1, A = 0\}$ to the space $\{S = 1, A = 1\}$.

\subsection{Estimation and Implementation}\label{subsec:syntharea_steps_method}

The sample analog of the theoretical result in Proposition~\ref{prop:synthArea-identification} is our proposed estimator for partial pooling by weighting, which we call the \emph{synthetic area} estimator:
\begin{equation*}
\widehat{\tau}^{\texttt{SA}}_{j} = \sum^{n}_{i=1}\widehat{w}^{\texttt{SA}}_{ij}Y_{i}
\end{equation*}
where the summation is over the survey sample and the crucial weight is given by
\begin{equation} \label{eq:sa-weight}
\widehat{w}^{\texttt{SA}}_{ij} \propto
\underbrace{\frac{\widehat{\Pr}(A_{i} = j\mid  \*X^{P}_{i}, \*X^{S}_{i}, S_{i} = 1)}{\widehat{\Pr}(A_{i} = j \mid \*X^{P}_{i}, S_{i} = 1)}}_{\widehat\zeta_{ij}}
\underbrace{ %
\vphantom{ \left(\frac{\widehat{\Pr}()}{\widehat{\Pr}()}\right) }
\Pr(A_{i} = j \mid \*X^{P}_{i})}_{p_{ij}} \frac{1}{\widehat{\Pr}(S_{i} = 1 \mid \*X^{P}_{i})}.
\end{equation}
As is clear from the above expression,
the proposed estimator places non-zero weight on all observations,
whereas the direct estimator (Equation~\eqref{est:direct}) puts weights of zero on observations outside of the target area.

The weights (Equation \ref{eq:sa-weight}) requires the estimation of three quantities to construct weights.
First, the sampling probability $\Pr(S_{i} = 1 \mid \*X_{i})$ must be estimated from available data.
However, this is an easier task than sampling probability for the direct estimator because we do not condition on the specific area.
Therefore, it has to be estimated once, and does not require the area specific data.

In fact, if the overall poll comes with a set of weights the pollster has already computed, our method can directly incorporate it instead of estimating $\Pr(S_{i} = 1 \mid \*X_{i})$ altogether.
Surveys often come supplied with post-stratification weights estimated by the pollster through incorporating and balancing numerous relevant covariates. For example, recent YouGov's weights in the Cooperative Congressional Election Study are post-stratified on variables including news interest, religion, and state election results \citep{CCES2018}.
Denoting these weights as $w^{\texttt{national}}_{i} = 1 / \Pr(S_{i} = 1 \mid \*X_{i})$, one can simply compute the synthetic area weights for area $j$ as
\begin{equation} \label{eq:sa-weight-national}
\widehat{w}^{\texttt{SA}}_{ij} \propto
\widehat\zeta_{ij}\cdot p_{ij}\cdot w^{\texttt{national}}_{i} .
\end{equation}

Second, we estimate the ratio $\zeta_{ij}$ of the probability $\Pr(A_{i} = j \mid \*X^{P}_{i}, \*X^{S}_{i})$ over $\Pr(A_{i} = j \mid \*X^{P}_{i})$
with survey samples. This quantity can be interpreted as a kind of similarity score assigned to outside observations, where observations similar to the area of interest are upweighted. The probabilities can be estimated by the logistic regression or the multinomial logit. To avoid separation issues with highly predictive covariates, in practice we implement a regression which induces shrinkage on the coefficients \citep[][\textsf{bayesglm} in \textsf{R}]{gelman2008weakly}.
Note that this step is done only with the survey data, and does not require the population data.

Finally, we denote the proportion of observations within each area given covariates, $\Pr(A_{i} = j \mid \*X^{P}_{i})$ as $p_{ij}.$
This quantity is computed with population data
and does not require the survey samples.
When covariates are discrete, this probability can be estimated nonparametrically.

We implement these steps in an open-source R package, \textsf{synthArea}. To summarize,
\begin{description}
  \item[Step 1] Estimate the sampling weights $\widehat{\pi}_{i} = \widehat{\Pr}(S_{i} = 1 \mid \*X^{P}_{i})$ and compute the population of proportion of the area $p_{ij} = \Pr(A_{i} = j \mid \*X^{P}_{i})$ for all $j$.
    Note that $\widehat{\pi}_{i}$ can be replaced by $\widehat{w}^{\texttt{national}}_{i}$.

  \item[Step 2] For each area $j$, estimate the ratio of the two probabilities
    using the survey data,
    \begin{equation*}
    \widehat{\zeta}_{ij} =
    \frac{\widehat{\Pr}(A_{i} = j \mid \*X^{P}_{i}, \*X^{S}_{i}, S_{i} = 1)}
         {\widehat{\Pr}(A_{i} = j \mid \*X^{P}_{i}, S_{i} = 1)}
    \end{equation*}
    where the numerator and the denominator are estimated via the multinomial logit or logistic regression.
  \item[Step 3] Estimate the target quantities by the synthetic area estimator by taking the weighted sum of all observations,
  \begin{equation*}
  \widehat{\tau}^{\texttt{SA}}_{j} = \sum^{n}_{i=1}\widehat{w}^{\texttt{SA}}_{ij}Y_{i},\quad
  \text{where}\quad
  \widehat{w}^{\texttt{SA}}_{ij} =
  \frac{\widehat{\zeta}_{ij}p_{ij} / \widehat{\pi}_{i}}{\sum^{n}_{i'=1}\widehat{\zeta}_{i'j}p_{i'j} / \widehat{\pi}_{i'}}.
  \end{equation*}

\end{description}

\subsection{Assessing the Conditional Independence Assumption}\label{sec:assess-independence}

The validity of a partial pooling approach including the proposed method hinges on whether the area ignorability (Assumption~\ref{assump:area-ignorability}) is plausible or not.
Although the assumption is imposed at the population level,  and so it is not directly testable, we can derive a testable condition for Assumption~\ref{assump:area-ignorability} (a).
Note that this is different from the MRP approach where the validity of partial pooling is usually
not directly testable as it is incorporated as a part of its model specification.

One implication of the conditional independence assumption is that the outcome is mean independent of $A_{i}$
if we condition on both types of covariates, $\*X^{P}_{i}$ and $\*X^{S}_{i}$,
\begin{equation*}
\E[Y_{i} \mid A_{i}, \*X^{P}_{i}, \*X^{S}_{i}, S_{i} = 1]
=
\E[Y_{i} \mid \*X^{P}_{i}, \*X^{S}_{i}, S_{i} = 1]
\end{equation*}
We can test if the above condition holds or not by assuming a model for the conditional expectation.
Suppose that we wish to assess if the assumption holds for the $j$th area.
Let $\widetilde{A}_{ij} = \bm{1}\{A_{i} = j\}$ denote an indicator that takes 1 if observation is from area $j$ and takes zero otherwise.
Then, assuming a linear model, we test if
the regression coefficient in front of $\widetilde{A}_{ij}$ is statistically distinguishable from zero.
\begin{equation*}
\E[Y_{i} \mid \widetilde{A}_{ij}, \*X^{P}_{i}, \*X^{S}_{i}, S_{i} = 1]
= \alpha + \delta_{j} \widetilde{A}_{ij} + \bm{\beta}^{\top}\*X^{P}_{i} + \bm{\gamma}^{\top}\*X^{S}_{i}
\end{equation*}
We recommend that scholars ideally use the equivalence test \citep{hartman2018equivalence,wellek2010testing} for testing if $\delta_{j} = 0$, as the failure to reject the conventional null $\delta_{j} = 0$ is not sufficient evidence for $\delta_{j} = 0$. Instead, the equivalence test posits the null of $|\delta_{j}| \geq \epsilon$ (i.e., the violation of the assumption by $\epsilon > 0$);
so the rejection of the null can be taken as
evidence for the validity of the assumption.

\section{Empirical Validation Study}\label{sec:validation}


For our main test of the method
we compare how well our weighting method recovers small area quantities in practice.
We demonstrate that the proposed method does as well as the standard MRP.

\subsection{Setup}

We use the Cooperative Congressional Election Study (CCES), attempting to estimate quantities at the congressional district (CD) level that can be compared with a ground truth. Validation studies using real data are critical for assessing the quality of a small area estimation, but such tests are still rare \citep[for exceptions, see][]{Warshaw2012,Hanretty2016a,fraga2020measuring}.
Studying congressional district level estimates from the CCES is a useful exercise because scholars use the CCES to infer small area quantities, but direct estimates are understood to be noisy \citep{broockman2017bias, Kalla2019}.
While the CCES weighted state samples are designed to be representative of each state's adult population,
the average number of respondents at the CD level is $148$. In contrast, each CD contains a population of about 750,000.

Our main outcome for validation is 2018 turnout as a proportion of the  Voting Age Population (VAP) in each congressional district. Turnout can be validated because district level turnout is observed with little measurement error. {In this application, we take the fraction of the turnout electorate provided by TargetSmart, divided by the estimated population 18 years and older in the ACS.} Turnout also has a desirable feature for validation in that a survey respondent's turnout in the CCES is measured with little misreporting. We use the validated vote of the CCES, which is an indicator of whether the survey respondent's personally identifiable information matches with the official voter file \citep{Ansolabehere2012}.  This means that there is limited misreporting error for our indicator, and any discrepancies between our estimate and the population turnout is more likely attributable to selection bias  instead of unrelated factors such as misreporting.

\subsection{Data}

For illustration we use the Florida and Texas subset of the 2018 CCES Common Content, totaling 8,854 respondents in 63 districts. The CCES is designed to match the national Voting Age Population through a matching procedure to YouGov's sampling frame. YouGov also provides post-stratification weights to weight samples to \emph{national} and \emph{state-level} populations, correcting for age, gender, education, race, and voter registration status at the national level and state-level election results \citep{CCES2018}.

We estimate our proposed method (outlined in Section \ref{subsec:syntharea_steps_method}) with the following predictors. For our population variables $\*X^P$, we use \texttt{sex}, binned \texttt{age}, and \texttt{education}. The joint distribution of these variables are taken from the ACS Table \texttt{B15001}.  While we are limited by what the ACS provides for $\*X^P$ (they do not release tables for more than three-way interactions at the CD-level), we can choose the survey variables $\*X^{S}$ freely. To meet Assumption \ref{assump:area-ignorability} we must include variables that we expect to be predictive of turnout and can explain the differences in turnout by CD. We use the three-way interaction of \texttt{race}, \texttt{news interest}, \texttt{partisan identification}.

\subsection{Estimation}

Our implementation follows the three components of our estimator (Proposition~\ref{prop:synthArea-identification}):
\begin{enumerate}
\item For $\pi_i = {\Pr}(S_{i} = 1 \mid \*X^{P}_{i})$,
we use the existing post-stratification weights as-is. That is, we use  $\widehat{\pi}_i \propto 1 / w_i^{\texttt{yougov}}$ where $w_i^{\texttt{yougov}}$ is the post-stratification weight provided by YouGov. Because YouGov's weights adjust for not only $\*X^P$ but also other important variables such as registration and news interest, one would expect that this weight to be incorporate more information.

\item For $p_{ij} = \Pr(A_{i} = j \mid \*X^{P}_{i})$, we only use the ACS (population) counts and simply count the fraction of areas given each cell of population covariates.

\item We use the survey variables in estimating the ratio $\widehat{\zeta}_{ij} = \widehat{\Pr}(A_{i} = j \mid \*X^{P}_{i}, \*X^{S}_{i}, S_{i} = 1)  /\widehat{\Pr}(A_{i} = j \mid \*X^{P}_{i}, S_{i} = 1)$. We estimate the numerator and denominator from the in-sample fitted probabilities of separate logistic regressions with shrinkage.
\end{enumerate}

We compare our estimates to a standard MRP model to examine the implications of each modeling approach. When modeling MRP, use the same survey data and post-stratification targets as much as possible. We include the same variables $\*X^P$ we include in our proposed synthetic estimator (\texttt{sex}, binned \texttt{age}, and \texttt{education}), and partially pool by modeling these as random effects. We also model random intercepts (Equation \ref{eq:mrp-area-re}) for each congressional district. We estimate the models with a Bayesian model with standard Normal priors, using the \textsf{brms} package.

\subsection{Assessing Area Ignorability}

Before comparing estimates, we first test whether the area-ignorability assumption is reasonable in turnout. For a given target $j$, We estimate the linear model regressing the outcome $Y_i$ on the survey variables $\*X^P$ and  $\*X^{S}$, and an indicator variable $\bm{1}\{A_i = j\}$ and simply examine the coefficient on the last indicator variable. Rejecting the null hypothesis that its coefficient is zero indicates a violation of Assumption \ref{assump:area-ignorability}.

\begin{figure}[tb]
\centering

\centerline{
  \includegraphics[width=\textwidth]{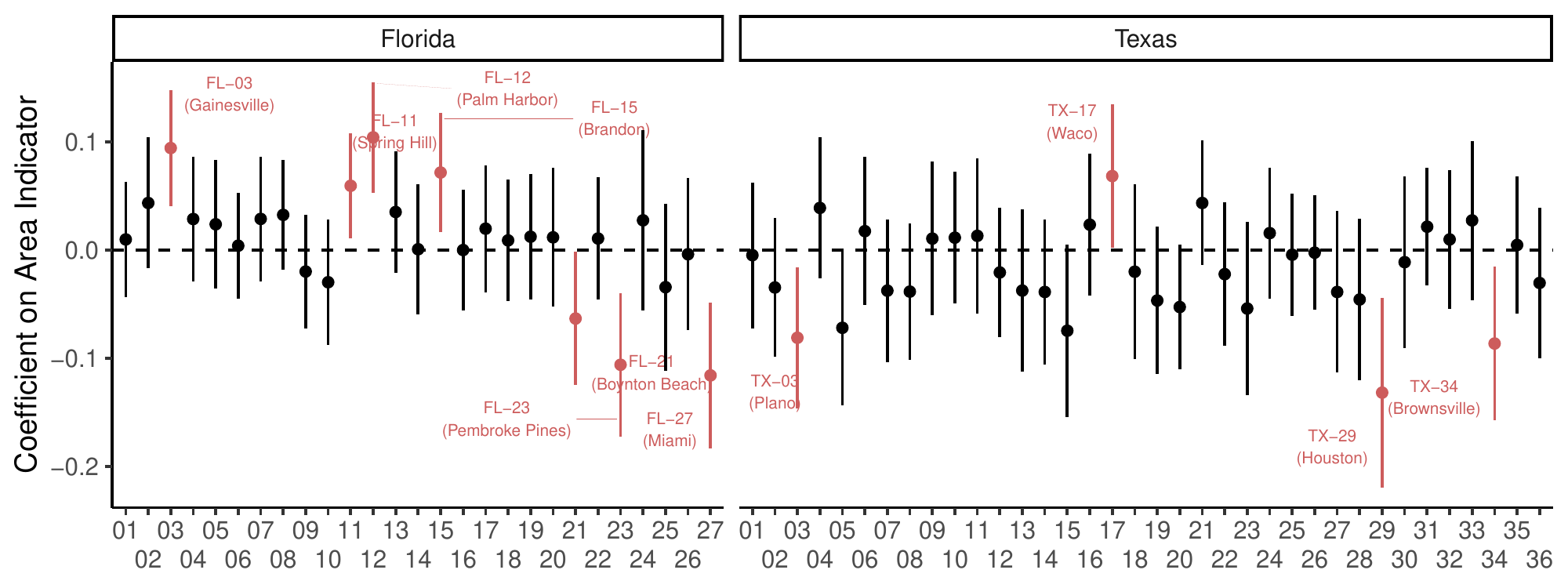}
}
\caption{\textbf{Test for Area Ignorability}.
  Each point shows the 95 percent confidence interval from the test statistic for a given congressional district. Districts whose confidence intervals do not cross zero are highlighted, with the name of the largest place in the district.
  }\label{fig:val_area-ignorability}
\end{figure}

Figure \ref{fig:val_area-ignorability} shows the range of the coefficients with 95 percent confidence intervals. Most district estimates are not distinguishable from zero, indicating that the condition for in-sample area ignorability is met.  Positive coefficients indicate that samples of the district of interest have higher turnout in the sample than observations matched on the covariates, and negative coefficients indicate the opposite.

Before turning to validation results, it is worth noting what this test can and cannot diagnose. Assumption \ref{assump:area-ignorability} is about the ignorability of areas in the sample. It does \emph{not} ask whether samples are representative of the population. Indeed, as we will see in the validation, districts that pass the area ignorability test as shown above may still overestimate turnout because Assumption \ref{assump:sampling-ignorability} validation is violated and the survey over-samples likely voters. That said, the test is still useful to respond to concerns that the small area estimation is \emph{over-pooling}.

\subsection{Comparison of Methods Estimating Turnout}

Our main results are in Figure \ref{fig:val_turnout-main}. Each scatter plot compares the true population turnout on the x-axis and the estimates on the y-axis.  We compute four sets of fit statistics and print them in the bottom-right of each figure: the root mean square error, the mean of the absolute value of the error, the mean of the (positive or negative) error, and the correlation. All values are taken across the 428 congressional districts we use.

\begin{figure}[tb]
\vspace{-1em}
\center

\includegraphics[width = \linewidth]{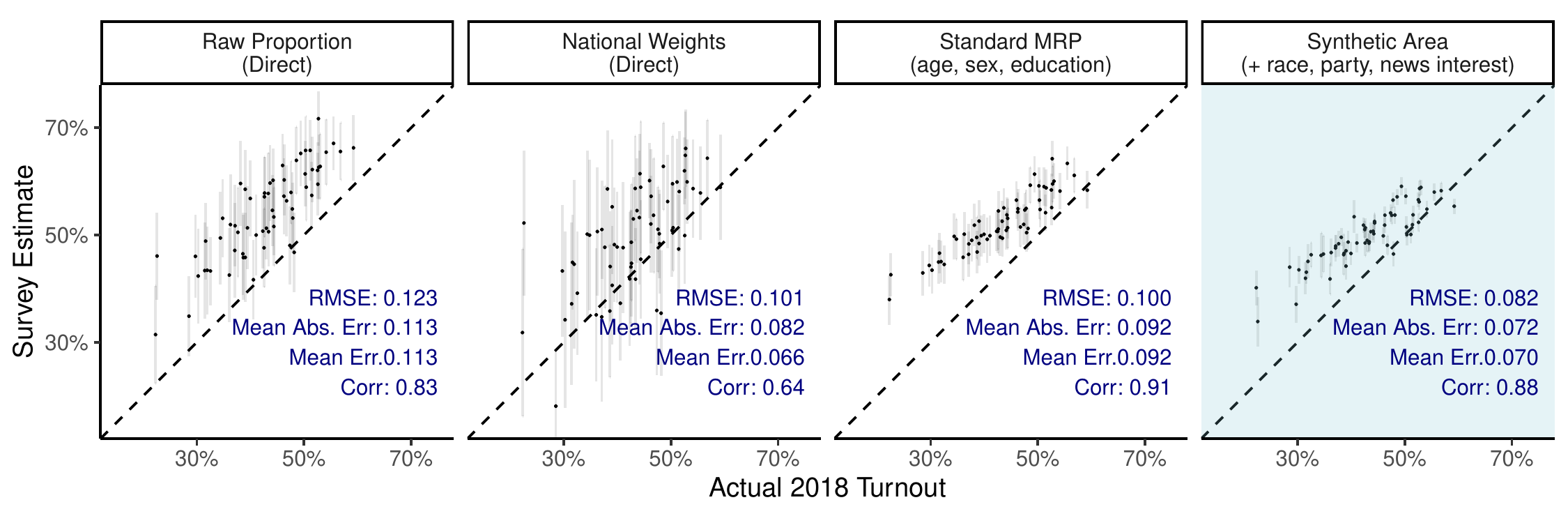}

\caption{\textbf{Predictive Performance of MRP and the Synthetic Area Estimator}.
Each scatter plot shares the same ground truth values of the x-axis and shows different model estimates, with 90 percent confidence intervals (direct estimator and synthetic area) or credible intervals (MRP), on the y-axis. Our proposed small area estimator is highlighted in blue. }
\label{fig:val_turnout-main}
\end{figure}

The quality of the estimates is measured by the four metrics presented at the corner of each panel.
The unweighted estimates (first panel) are correlated with actual turnout, but prediction error is on the order of 10 percentage points. The survey consistently overestimates turnout rates.  The post-stratification weights --- which are designed to make weighted samples of the entire CCES or state-level estimates representative --- do reduce prediction error, for example by down-weighting higher education voters who are more likely to turn out.

The MRP estimates (third panel) improve the prediction error compared to the direct estimates (first two panels). Partially pooling CD-level intercepts leads to a  MRP in this example appears to over-pool towards the global mean. MRP estimates have  smaller cross-district variation, and the estimates tend to align cleanly with the ground truth. Ideally, we would include more demographic variables to improve the outcome model, but this is limited by the data available in the CD-level population distributions. Synthetic \emph{population} methods \citep{leemann2017extending} provide some tools to do this which would amount to expanding the set of population variables $\*X^{P}$, benefitting both MRP and our estimator.

Finally, our synthetic area estimator that partially pools with more variables reduces the prediction error and is comparable with, if not more effective than, MRP estimates. Our estimates (highlighted in light blue), pulls estimates closer towards the ground truth and decreases the RMSE to 0.082, or 8.2 percentage points. The distribution of points is similar to MRP, but amounts to an improvement over MRP in the aggregate. Recall that our main differences between MRP is that we can adjust for survey-only variables such as news interest, race, and party identification. We also use YouGov's post-stratification weights as a pre-estimated propensity score for selection, thereby indirectly using other information that YouGov adjusted for.  In summary, we interpret these results as representing the value of adjusting for more relevant covariates while relaxing the level of the post-stratification.

Still, all the models in Figure \ref{fig:val_turnout-main} tend to over-estimate turnout. Indeed, this example is useful to highlight the distinction between Assumption \ref{assump:sampling-ignorability} and \ref{assump:area-ignorability} because the CCES is generally known to over-represent likely voters, even after standard post-stratification. The partial pooling component of small area estimation, whether it is MRP or our estimator, functions orthogonally from the question of selection. Selection is handled in the post-stratification step, which is roughly comparable across the specifications in this Figure.

To the extent that our method errs in predicting turnout, it errs in similar ways as MRP. Figure \ref{fig:corr_error} compares model estimates against each other, after subtracting out the true turnout to highlight model error. We compare the prediction error of our synthetic area estimator to the prediction error of the raw average and MRP. Our error is more strongly correlated with the prediction errors of MRP, at about 0.89. This is consistent with the idea that both models are essentially models of partially pooling across areas. We also see from the outliers in Figure \ref{fig:corr_error} that prediction errors are especially acute in heavily Hispanic CDs, consistent with studies that show consistent underestimation of minority turnout \citep{fraga2021CPS}.

\begin{figure}[tb]
\centering
\includegraphics[width=\linewidth]{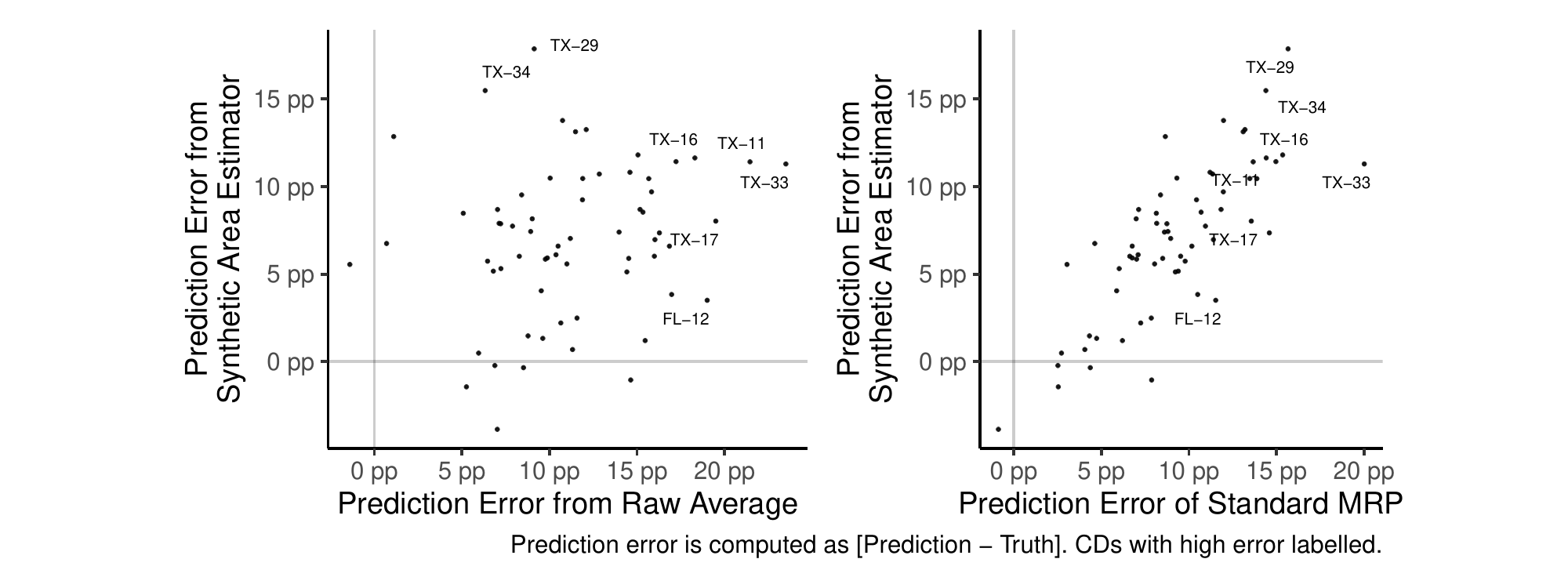}
\caption{
  \textbf{Top panel: Correlation Between Model Errors}.
  Predictions come from the same values in Figure~\ref{fig:val_turnout-main}.
}
\label{fig:corr_error}

\end{figure}

\section{Empirical Application}\label{sec:application}


To illustrate the assumptions underlying the synthetic area approach and its potential uses, we turn again to the CCES 2018 data. Instead of turnout, we estimate the support for immigration policies at the congressional district level, which can then be used as predictors and outcomes for studying how U.S. House members represent their district opinion.
In this application, we demonstrate the importance of survey variables to obtain reliable estimates in small areas.
Florida is a large and diverse state, with six out of its 27 congressional districts being majority Black or majority Hispanic. Hispanic groups in Florida are also a mix of those with Mexican, Cuban, Puerto Rican, and Venezuelan ancestry. We would expect there to be rich district-level heterogeneity in demographics and therefore preferences.

\begin{figure}[!tb]
  \centerline{
    \includegraphics[width=\linewidth]{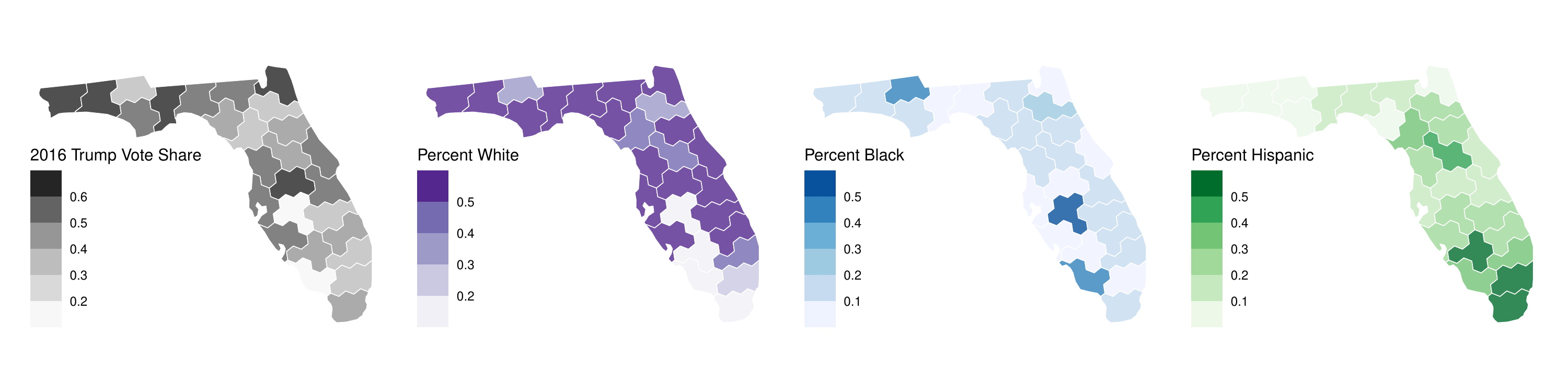}
  }
  \caption{\textbf{Congressional Districts in Florida}. Darker color shows  higher rate of support for Trump in the 2016 election (the first column), and higher proportion of each racial group.}
  \label{fig:fl_demographics}

\includegraphics[width = \linewidth]{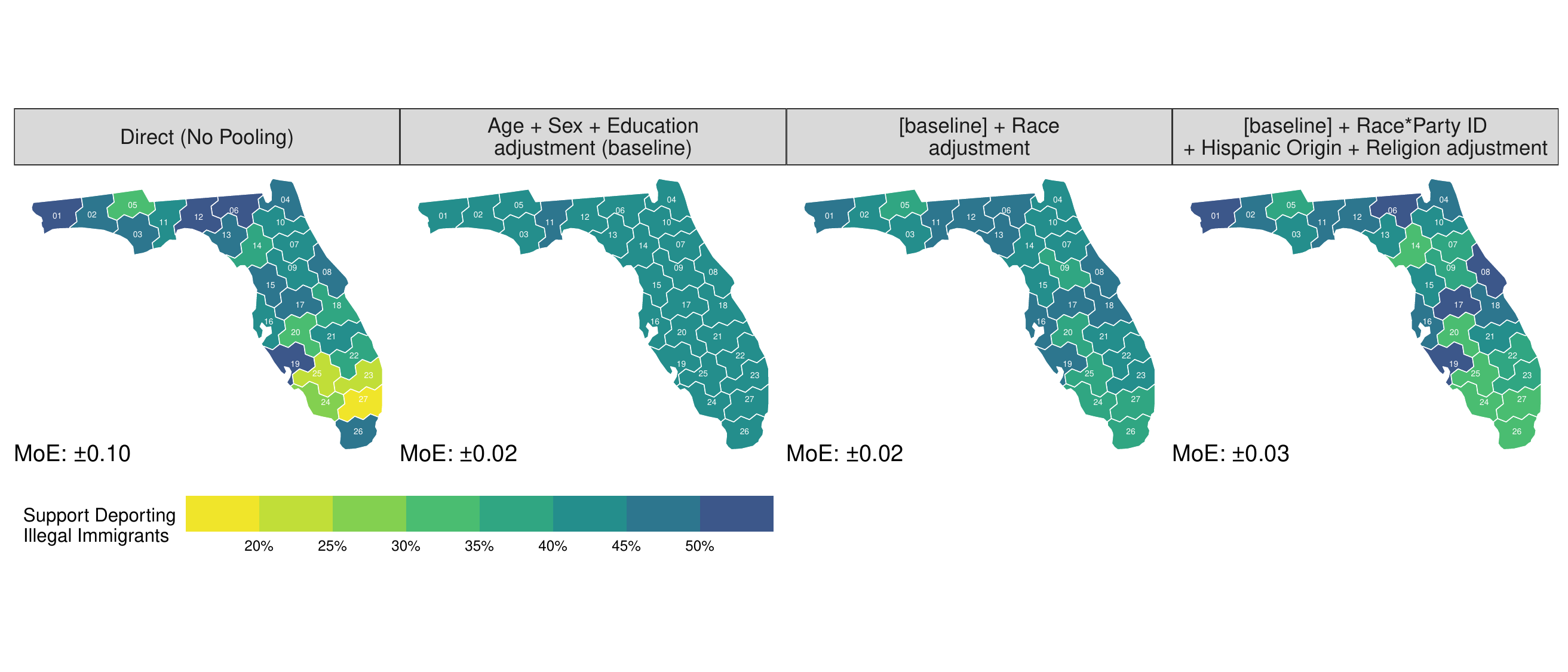}
\caption{\textbf{Estimated support for deportation at the CD level}. The Direct Estimate panel shows simple averages. The remaining three panels show estimates from the synthetic estimator, in increasing levels of covariate adjustment. The simplest version with no survey-only variables appears to over-homogenize opinion, while adjusting for differences in race, Hispanic origin, party identification, and religion distinguishes districts in reasonable ways. The average margin of error (twice the standard error) for a given district estimate is consistently lower for synthetic area estimates due to partial pooling.}
\label{fig:fl_deportation_res}
\end{figure}

Figure~\ref{fig:fl_demographics} shows some background demographics of each congressional district, taken from the ACS. To give equal size to roughly equal sized congressional districts, we use the cartogram and area descriptions developed by Daily Kos \citep{donnermap}. This shows the variation in dense areas such as Miami  better than a normal map, but it does distort the place of locations within the state somewhat

We applied the proposed method to the CCES question that asks respondents about their attitude toward immigration policies.
Specifically, we use a question that asks if the US government should, in the language of the survey, ``identify and deport illegal immigrants'' (\texttt{CC16\_331\_7}).
We use a binary coding of this question, using as 1 if a respondent agrees with the statement, and 0 otherwise.
Our quantity of interest is therefore the proportion of adults who support the deportation policy at each congressional district.

We compare three different sets of covariate adjustments against the direct estimator that uses only observations from each district to estimate the quantity of interest.
As a baseline estimate, we adjust for basic demographic variables (\texttt{age}, \texttt{sex}, and the level of \texttt{education}).
As a second model, we adjust for \texttt{race} which is observed in the survey but is not observed at the population level.
In the last model, we adjust for a variety of individual level characteristics that are observed only in the survey, such as respondent's party identification,
 and responses to the other immigration questions in the battery. We also control for whether a Hispanic respondent identifies as being of Cuban, Puerto Rican, Mexican, or other South American descent, each as binary indicator variables.
We emphasize that these variables have not been utilized in the existing small area estimation methods, because none of them are observed at the population level.

Figure~\ref{fig:fl_deportation_res} shows the estimated percent of support for deporting illegal immigrants across four sets of models. The first panel starts with direct effects which exhibit high variation, but in somewhat unreliable or unexpected ways.
The two southern districts in the Miami area are both majority Hispanic, but they show noticeably different estimated levels of support for this group. In any case, the margin of area for each area estimate is 10 percentage points on average, making its use impractical.

The second panel uses the proposed method but only with three population variables (i.e., $\*X^{P}$) for poststratification without any survey-only variables.
In contrast, it shows that estimates are similar across the state. All but one CD has an estimate in the 40-45\% range, which may strike someone as too homogeneous. This example shows the pitfalls of adjusting for too few variables in partial pooling.

The third panel shows the estimate with one additional survey variable: \texttt{race}.
By accounting for the within-survey racial heterogeneity, we get more diversity in estimates compared to the second panel that does not account for the survey variables.

In the final panel, we add an even richer set of covariate adjustments and obtain an even more heterogeneous set of estimates that retain a reasonable standard error. We include interactions for race and party identifications, allowing, for example, for partisan divergences to differ between Whites and Hispanics.  We include binary variables for Hispanic origin, distinguishing, for example, between Cuban and Mexican Americans.

In this application, we do not have access to the ground truth.
Thus, unlike in the previous section, we cannot necessarily validate these estimates.
We can still assess if the estimates are reasonable by checking if districts with similar demographics exhibit similar estimates.
The Southwest tip of the Florida peninsula shows the Miami area, roughly Congressional Districts 20, 21, 22, 23, 25, and 27. We can see that in the most extensive model (the fourth column in the figure),
the estimates are consistent with our expectations that race should shape views for immigration policy:
The Hispanic districts FL-23 (Broward County - center) and FL-27 (Southern Miami area and Coral Gables) have similar levels of opposition to the deportation policy,
the Black districts FL-20 (West Palm Beach and Fort Lauderdale areas) and FL-24 (Northern Miami) have similar levels, and the majority White districts of FL-21 (Southeast Palm Beach County) and FL-22 (Boca Raton and Northeast Broward County) show similar levels too.

In sum, the synthetic area estimator produces reasonable estimates of
the support for deportation that correlate with the racial composition of the district,
but only after adjusting for a sufficient set of covariates. The importance of covariates should be expected from the two assumptions for small area estimation we have shown.


\FloatBarrier

\section{Conclusion}


Small area estimation is practical in spirit: it allows survey researchers to stabilize subnational estimates that would be prohibitively noisy on their own. This benefit, however, comes with the assumption required for borrowing information across areas. Somewhat surprisingly, the concept of ``borrowing information'', while widely referenced, has rarely been formally defined upfront as an identification assumption, even in extensive simulation validations of MRP \citep{buttice2013does}. We suspect one reason for this ambiguity is because outcome modeling approaches have been dominant in small area estimation methods.

In this paper we proposed a small area estimation method in a weighting approach. Unlike classic small area estimation \citep{fay1979estimates}, our \emph{synthetic area} estimator partially pools by combining the direct estimate with an indirect estimate that is weighted by a propensity score to ``look like'' the area of interest.  We highlighted several practical benefits of this choice.  Mainly, this allows researchers to incorporate individual-level predictors of area heterogeneity that are out of reach for random effects approaches, which can only incorporate area-level aggregates.

Our synthetic estimator approach aligns even more closely to MRP methods. The key innovation of MRP over classic small area estimation was to correct outputs of a random effects model for unrepresentativeness and selection bias \citep{park2004bayesian}. Our synthetic estimator achieves the same goal by combining the selection ignorability assumption (Assumption \ref{assump:sampling-ignorability}) with the area ignorability assumption (Assumption \ref{assump:area-ignorability}) in a single estimator.

How should applied researchers choose between these methods?  MRP, like classic small area estimation, limits itself to variables that are measured in the population or area-level variables.  Nevertheless MRP may be preferable if the researcher is focused on one outcome and is willing to build a massive model with deep interactions where principles of shrinkage \citep{stein1981estimation} embodied in random effects are likely to pay off. In contrast, our approach can be better suited when researchers and pollsters want to estimate a set of weights that can be used to any other outcome in the survey. Our method also better integrates the advances in causal inference for observational data in recent years \citep{hazlett2016kernel,wong2018kernel, hainmueller2012entropy} because each step of propensity score computation can be easily replaced with any of these new estimators. We also note that it is possible in future work to combine our estimator with the MRP to form a doubly robust estimator that allows for the misspecification of either the sampling model or outcome model \citep[e.g.,][]{tan2010bounded, robins1995analysis}. In fact, such estimators are increasingly popular to correct inference when models are estimated with machine learning methods \citep{chernozhukov2018double}.

Our theoretical proposition in this paper is relevant whether one uses our method or an outcome-based approach like MRP. We essentially suggest that researchers can ask two questions when they evaluate the validity of any small area estimate: First, are a sufficient number of relevant covariates included in the post-stratification portion such that selection bias is ignorable? And second, are a sufficient number of relevant covariates included in both the post-stratification and pooling components, such that area or subgroup boundaries are ignorable? 

In practice these conditions are rarely met. Selection bias that cannot be accounted for by standard post-stratification is a problem for all modern survey research, whether small area estimates or not.  And when it comes to area ignorability, scholars might have substantive reasons to be skeptical of claims that a synthetic California, for example, is exchangeable with survey responses actually from California. Our hypothesis test can detect whether such area exchangeability assumptions are plausible. But generally speaking, identification results do not come for free in a statistical model: they must be bought with data. One task for pollsters and researchers, then, is to collect and measure covariates that are sufficiently rich to minimize the double selection problem we have highlighted.

\clearpage
\appendix

\section{Proofs}\label{sec:proofs}


\begin{lemma}[Useful Properties]\label{lemma:useful}
Let $X$, $Y$, and $Z$ be random variables.
Also let $D \in \{0, 1\}$ be a binary random variable.
Then, following equalities hold:
\begin{itemize}
  \item[(a)] $\E\big[\E(X \mid Z) Y \big] = \E\big[X \E(Y \mid Z) \big]$
  \item[(b)] $\E[ Y \mid X, D = 1] = \E[DY / \Pr(D = 1 \mid X)\mid X]$
\end{itemize}
\end{lemma}


\subsection{Proof of Proposition~\ref{prop:synthArea-identification}}
Under Assumption~\ref{assump:sampling-ignorability} and \ref{assump:area-ignorability}, the estimand $\tau_{j}$ is nonparametrically identified as
\begin{align*}
\tau_{j} = \E\bigg\{
\frac{\bm{1}\{S = 1\}}{\Pr(A = j)}
\frac{\Pr(A = j\mid  \*X^{P}, \*X^{S}, S = 1)}{\Pr(A = j \mid \*X^{P}, S = 1)}
\frac{\Pr(A = j \mid \*X^{P})}{\Pr(S = 1 \mid \*X^{P})}
 Y
\bigg\}
\end{align*}
for  $j = 1, \ldots, J$.

\begin{proof}

\begin{align}
\E(Y \mid A = j)
&=
\E\Big \{
  \E(Y \mid \*X^{P}, S = 1, A = j)
  \mid A = j
\Big\} \nonumber \\
&=
\E\Big \{
  \E\big[ \E(Y \mid \*X^{P}, \*X^{S}, A = j, S = 1) \mid \*X^{P}, S = 1, A = j \big]
  \ \big|\ A = j
\Big\} \nonumber \\
&=
\E\bigg\{
\frac{\bm{1}\{A = j\}}{\Pr(A = j)}
\E\bigg[
  \frac{\bm{1}\{S = 1, A = j\}}{\Pr(S = 1, A = j \mid \*X^{P})}
  \E(Y \mid \*X^{P}, \*X^{S}, S = 1) \ \big|\ \*X^{P}
\bigg]
\bigg\}\nonumber \\
&=
\E\bigg\{
\frac{\Pr(A = j\mid \*X^{P})}{\Pr(A = j)}
\bigg[
  \frac{\bm{1}\{S = 1, A = j\}}{\Pr(S = 1, A = j \mid \*X^{P})}
  \E(Y \mid \*X^{P}, \*X^{S}, S = 1)
\bigg]
\bigg\}\label{eq:proof-prop1}
\end{align}
where the first equality is due to the law of iterated expectation
and Assumption~\ref{assump:sampling-ignorability}, the second equality is due to the law of iterated expectation and Assumption~\ref{assump:area-ignorability},
and the third and the fourth equality are the direct application of Lemma~\ref{lemma:useful}.

Now, we focus on the expectation term inside of the outer expectation,
\begin{align*}
\E& (Y \mid \*X^{P}, \*X^{S}, S = 1)\\
&= \E(Y \mid \*X^{P}, \*X^{S}, S = 1, A = j)\Pr(A = j\mid  \*X^{P}, \*X^{S}, S = 1)\\
&\qquad  +
\E(Y \mid \*X^{P}, \*X^{S}, S = 1, A \neq j)\Pr(A \neq j\mid  \*X^{P}, \*X^{S}, S = 1)\\
&= \E\bigg\{
\frac{\bm{1}\{S = 1, A = j \} }{\Pr(S = 1, A = j \mid \*X^{P}, \*X^{S})} Y\ \big|\ \*X^{P}, \*X^{S}
\bigg\} \Pr(A = j\mid  \*X^{P}, \*X^{S}, S = 1)\\
&\qquad +
\E\bigg\{
\frac{\bm{1}\{S = 1, A \neq j \} }{\Pr(S = 1, A \neq j \mid \*X^{P}, \*X^{S})} Y\ \big|\ \*X^{P}, \*X^{S}
\bigg\} \Pr(A \neq j\mid  \*X^{P}, \*X^{S}, S = 1)
\end{align*}
where the first equality is due to the law of iterated expectation
and the final line is again the direct application of Lemma~\ref{lemma:useful}.
Then, we can express the terms in the square brackets as
\begin{align*}
&
  \frac{\bm{1}\{S = 1, A = j\}}{\Pr(S = 1, A = j \mid \*X^{P})}
  \E(Y \mid \*X^{P}, \*X^{S}, S = 1)
\\
&=
\frac{\bm{1}\{S = 1, A = j\}}{\Pr(S = 1, A = j \mid \*X^{P})}
\E\bigg\{
\frac{\bm{1}\{S = 1, A = j \} }{\Pr(S = 1, A = j \mid \*X^{P}, \*X^{S})} Y\ \big|\ \*X^{P}, \*X^{S}
\bigg\} \Pr(A = j\mid  \*X^{P}, \*X^{S}, S = 1)
\\
&\quad +
\frac{\bm{1}\{S = 1, A = j\}}{\Pr(S = 1, A = j \mid \*X^{P})}
\E\bigg\{
\frac{\bm{1}\{S = 1, A \neq j \} }{\Pr(S = 1, A \neq j \mid \*X^{P}, \*X^{S})} Y\ \big|\ \*X^{P}, \*X^{S}
\bigg\} \Pr(A \neq j\mid  \*X^{P}, \*X^{S}, S = 1).
\end{align*}

Finally, we plug in the above expression to Equation~\eqref{eq:proof-prop1}
and obtain,
\begin{align*}
\E&\bigg\{
\frac{\Pr(A = j\mid \*X^{P})}{\Pr(A = j)}
\bigg[
  \frac{\bm{1}\{S = 1, A = j\}}{\Pr(S = 1, A = j \mid \*X^{P})}
  \E(Y \mid \*X^{P}, \*X^{S}, S = 1)
\bigg]
\bigg\}\\
&=
\E\bigg\{
\frac{\Pr(A = j\mid \*X^{P})}{\Pr(A = j)}
\frac{\bm{1}\{S = 1, A = j\}}{\Pr(S = 1, A = j \mid \*X^{P})}
\E\bigg[
\frac{\bm{1}\{S = 1, A = j \} }{\Pr(S = 1, A = j \mid \*X^{P}, \*X^{S})} Y\ \big|\ \*X^{P}, \*X^{S}
\bigg] \\
&\qquad\qquad \times \Pr(A = j\mid  \*X^{P}, \*X^{S}, S = 1)
\bigg\} \\
& +
\E\bigg\{
\frac{\Pr(A = j\mid \*X^{P})}{\Pr(A = j)}
\frac{\bm{1}\{S = 1, A = j\}}{\Pr(S = 1, A = j \mid \*X^{P})}
\E\bigg[
\frac{\bm{1}\{S = 1, A \neq j \} }{\Pr(S = 1, A \neq j \mid \*X^{P}, \*X^{S})} Y\ \big|\ \*X^{P}, \*X^{S}
\bigg] \\
&\qquad \qquad \times \Pr(A \neq j\mid  \*X^{P}, \*X^{S}, S = 1)
\bigg\} \\
&=
\E\bigg\{
\frac{\Pr(A = j\mid \*X^{P})}{\Pr(A = j)}
\frac{\bm{1}\{S = 1, A = j\}}{\Pr(S = 1, A = j \mid \*X^{P})}
 Y
\Pr(A = j\mid  \*X^{P}, \*X^{S}, S = 1)
\bigg\} \\
&\quad +
\E\bigg\{
\frac{\Pr(A = j\mid \*X^{P})}{\Pr(A = j)}
\frac{\bm{1}\{S = 1, A \neq j\}}{\Pr(S = 1, A = j \mid \*X^{P})}
 Y
\Pr(A = j\mid  \*X^{P}, \*X^{S}, S = 1)
\bigg\} \\
&=
\E\bigg\{
\frac{\bm{1}\{S = 1\}}{\Pr(A = j)}
\frac{\Pr(A = j\mid  \*X^{P}, \*X^{S}, S = 1)}{\Pr(A = j \mid \*X^{P}, S = 1)}
\frac{\Pr(A = j\mid \*X^{P})}{\Pr(S = 1 \mid \*X^{P})}
 Y
\bigg\}
\end{align*}
which completes the proof.
\end{proof}

\subsection{Proof of Corollary~\ref{cor:alternative-identification}}
Let $\pi_{j}(\*X) = \Pr(S = 1 \mid \*X^{P}, A = j)$ and $p_{j}(\*X) = \Pr(A = j \mid \*X^{P}, \*X^{S}, S = 1)$.
Then, the target estimand is also identified as
\begin{align*}
\tau_{j} &=
\E\bigg\{
\frac{\bm{1}\{S = 1, A = j \}}{\Pr(A = j)}
\frac{Y}{\pi_{j}(\*X)}\
p_{j}(\*X)
\bigg\}\\
&\quad +
\E\bigg\{
\frac{\bm{1}\{S = 1, A \neq j \}}{\Pr(A = j)}
\frac{\Pr(A = j \mid \*X^{P}, \*X^{S}, S = 1)}{\Pr(A \neq j \mid \*X^{P}, \*X^{S}, S = 1)}
\frac{Y}{\pi_{j}(\*X)}\
(1 - p_{j}(\*X))
\bigg\}.
\end{align*}

\begin{proof}
From the last part of the proof in Proposition~\ref{prop:synthArea-identification}, we have that
\begin{align*}
\tau_{j} &=
\E\bigg\{
\frac{\Pr(A = j\mid \*X^{P})}{\Pr(A = j)}
\frac{\bm{1}\{S = 1, A = j\}}{\Pr(S = 1, A = j \mid \*X^{P})}
 Y
\Pr(A = j\mid  \*X^{P}, \*X^{S}, S = 1)
\bigg\} \\
&\quad +
\E\bigg\{
\frac{\Pr(A = j\mid \*X^{P})}{\Pr(A = j)}
\frac{\bm{1}\{S = 1, A \neq j\}}{\Pr(S = 1, A = j \mid \*X^{P})}
 Y
\Pr(A = j\mid  \*X^{P}, \*X^{S}, S = 1)
\bigg\}
\end{align*}

Since $\Pr(S = 1, A = j \mid \*X^{P}) = \Pr(S = 1 \mid \*X^{P}, A = j)\Pr(A = j \mid \*X^{P})$, the first term of the right-hand side in the above display
is given by
\begin{equation*}
\E\bigg\{
\frac{\bm{1}\{S = 1, A = j\}}{\Pr(A = j)}
\frac{Y}{\Pr(S = 1 \mid \*X^{P}, A = j)}\ p_{j}(\*X)
\bigg\}.
\end{equation*}
Similarly, by multiplying the second term by
$\Pr(A \neq j \mid \*X^{P}, \*X^{S}, S = 1) / \Pr(A \neq j \mid \*X^{P}, \*X^{S}, S = 1)$, we get
\begin{equation*}
\E\bigg\{
\frac{\bm{1}\{S = 1, A \neq j\}}{\Pr(A = j)}
\frac{Y}{\Pr(S = 1 \mid \*X^{P}, A = j)}
\times \frac{\Pr(A = j \mid \*X^{P}, \*X^{S}, S = 1)}{\Pr(A \neq j \mid \*X^{P}, \*X^{S}, S = 1)}\
(1 - p_{j}(\*X))
\bigg\}
\end{equation*}

Combining the above two expectation, we obtain the result.
\end{proof}

%

\setstretch{1.1}
\clearpage
\bibliography{synthArea}

\end{document}